\documentclass[11pt,a4paper]{article}
\usepackage{fullpage}

\usepackage[english]{babel}

\usepackage{amssymb}
\usepackage[leqno]{amsmath}
\makeatletter
\newcommand{\leqnomode}{\tagsleft@true\let\veqno\@@leqno}
\newcommand{\reqnomode}{\tagsleft@false\let\veqno\@@eqno}
\makeatother

\usepackage{mathabx}
\usepackage{mathtools}
\usepackage{csquotes}
\usepackage[all]{nowidow}
\usepackage{amsthm,amsmath,amssymb}
\usepackage[mathcal,mathscr]{eucal}
\usepackage{epsfig,graphicx,graphics,color}
\usepackage{enumerate}
\usepackage[sort,nocompress]{cite}
\usepackage{hyperref}
\hypersetup{
    colorlinks=true,       % false: boxed links; true: colored links
    linkcolor=blue,        % color of internal links (change box color with linkbordercolor)
    citecolor=red,         % color of links to bibliography
    filecolor=magenta,     % color of file links
    urlcolor=cyan,         % color of external links
    linktocpage=true
}

\usepackage{algorithm}
\usepackage{bbm}
\usepackage[noend]{algpseudocode}
\usepackage{varwidth}
\newcounter{algsubstate}

\newenvironment{algsubstates}
  {\setcounter{algsubstate}{0}%
   \renewcommand{\State}{%
     \refstepcounter{algsubstate}%
     \Statex {\hspace{\algorithmicindent}\footnotesize\alph{algsubstate}:}\space}}
  {}
  
\theoremstyle{plain}
\newtheorem{theorem}{Theorem}
\newtheorem{lemma}[theorem]{Lemma}
\newtheorem{corollary}[theorem]{Corollary}
\newtheorem{claim}[theorem]{Claim}

\theoremstyle{definition}

\newcommand{\eps}{\varepsilon}

\newcommand{\TP}{\mathsf{TP}}
\newcommand{\TR}{\mathsf{T}}
\newcommand{\Eps}{\mathcal{E}}
\newcommand{\varT}{\mathcal{T}}
\newcommand{\varL}{\mathcal{L}}

\DeclareMathOperator\supp{supp}
\DeclareMathOperator\spa{span}
\DeclareMathOperator\comp{comp}

\DeclarePairedDelimiter{\floor}{\lfloor}{\rfloor}

\def\final{0}  % set this to 1 to get a comment-free version
\ifnum\final=0  %namely if we allow comments in the output
\newcommand{\knote}[1]{{\color{red}[{\tiny Krist\'of: \bf #1}]\marginpar{\color{red}*}}}
\newcommand{\rnote}[1]{{\color{blue}[{\tiny Roland: \bf #1}]\marginpar{\color{Blue}*}}}
\newcommand{\mnote}[1]{{\color{red}[{\tiny Matthias: \bf #1}]\marginpar{\color{red}*}}}
\newcommand{\anote}[1]{{\color{red}[{\tiny Andre: \bf #1}]\marginpar{\color{red}*}}}
\newcommand{\todo}[1]{{\color{red}[{\tiny TODO: \bf #1}]\marginpar{\color{red}*}}}
\else % in this case [final=1] we don't want any comments to show
\newcommand{\knote}[1]{}
\newcommand{\rnote}[1]{}
\newcommand{\mnote}[1]{}
\newcommand{\anote}[1]{}
\newcommand{\todo}[1]{}
\fi

\title{Degree-Bounded Generalized Polymatroids and\\ Approximating the Metric Many-Visits TSP\thanks{Supported by DAAD with funds of the Bundesministerium f{\"u}r Bildung und Forschung (BMBF) and by DFG project MN 59/4-1.}}

\author{Krist{\'o}f B{\'e}rczi\thanks{MTA-ELTE Egerv\'ary Research Group, Department of Operations Research, E{\"o}tv{\"o}s Lor{\'a}nd University, Hungary. Email: \texttt{berkri@cs.elte.hu}.} 
\and
Andr{\'e} Berger\thanks{Department of Quantitative Economics, Maastricht University, The Netherlands. Email: \texttt{a.berger@maastrichtuniversity.nl}.}
\and
Matthias Mnich\thanks{Universit{\"a}t Bonn \emph{and} Technische Universit{\"a}t Hamburg, Germany. Email: \texttt{matthias.mnich@tuhh.de}.}
\and
Roland Vincze\thanks{Department of Quantitative Economics, Maastricht University, The Netherlands \emph{and} Technische Universit{\"a}t Hamburg, Germany. Email: \texttt{roland.vincze@tuhh.de}.}
}

\begin{document}
\date{}
\maketitle

\begin{abstract}
  In the {\sc Bounded Degree Matroid Basis Problem}, we are given a matroid and a hypergraph on the same ground set, together with costs for the elements of that set as well as lower and upper bounds $f(\eps)$ and $g(\eps)$ for each hyperedge~$\eps$.
  The objective is to find a minimum-cost basis $B$ such that $f(\eps) \leq |B \cap \eps| \leq g(\eps)$ for each hyperedge $\eps$.
  Kir{\'a}ly et al. (Combinatorica, 2012) provided an algorithm that finds a basis of cost at most the optimum value which violates the lower and upper bounds by at most $2 \Delta-1$, where $\Delta$ is the maximum degree of the hypergraph.
  When only lower or only upper bounds are present for each hyperedge, this additive error is decreased to $\Delta-1$.
  
  We consider an extension of the matroid basis problem to generalized polymatroids, or g-po\-ly\-mat\-ro\-ids, and additionally allow element multiplicities.
  The {\sc Bounded Degree g-po\-ly\-mat\-ro\-id Element Problem with Multiplicities} takes as input a g-polymatroid $Q(p,b)$ instead of a matroid, and besides the lower and upper bounds, each hyperedge~$\eps$ has element multiplicities $m_\eps$.
  %In case of a single bound present, we call the problem {\sc Lower/Upper Bounded Degree Matroid Basis Problem with Multiplicities}.
  Building on the approach of Kir{\'a}ly et al., we provide an algorithm for finding a solution of cost at most the optimum value, having the same additive approximation guarantee.
  
  As an application, we develop a $1.5$-approximation for the metric {\sc Many-Visits TSP}, where the goal is to find a minimum-cost tour that visits each city $v$ a positive $r(v)$ number of times.
  Our approach combines our algorithm for the {\sc Bounded Degree g-polymatroid Element Problem with Multiplicities} with the principle of Christofides' algorithm from 1976 for the (single-visit) metric TSP, whose approximation guarantee it matches.
  \bigskip

\end{abstract}

\raisebox{-60ex}[0pt][0pt]{\hspace{78ex}\includegraphics[scale=0.45]{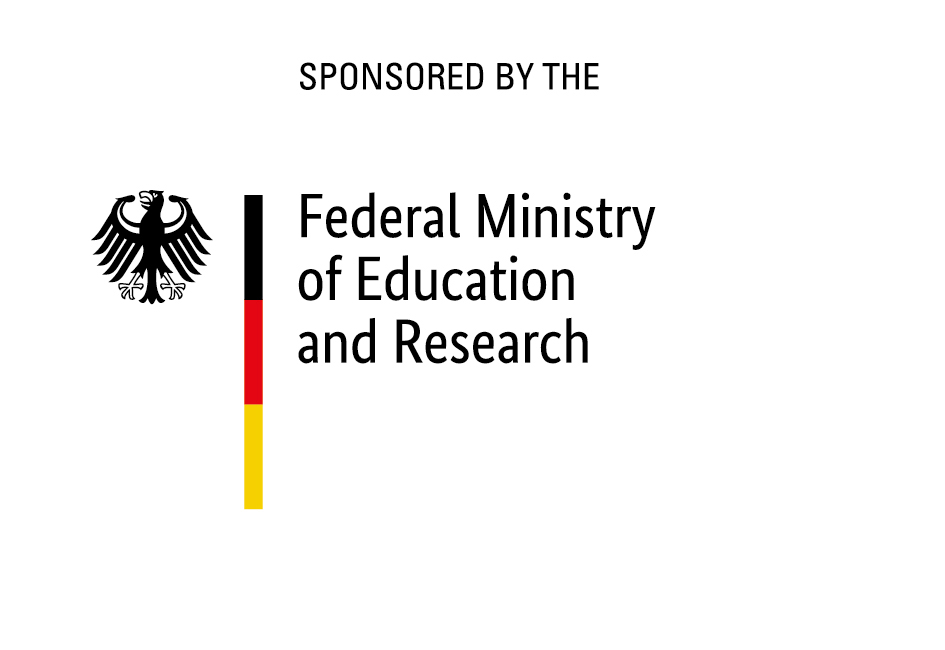}}

\thispagestyle{empty}

\clearpage
\pagebreak

\setcounter{page}{1}

\section{Introduction}
\label{sec:introduction}
In this paper we consider polymatroidal optimization problems with degree constraints.
An illustrious example is the {\sc Minimum Bounded Degree Spanning Tree problem}, where the goal is to find a minimum cost spanning tree in a graph with lower and upper bounds on the degree at each vertex.
Checking feasibility of a degree-bounded spanning tree contains the $\mathsf{NP}$-hard Hamiltonian path problem; therefore, efficiently finding spanning trees that only slightly violate the degree constraints, is of interest.
Several algorithms were given that were balancing between the cost of the spanning tree and the violation of the degree bounds~\cite{ChaudhuriEtAl2009,ChaudhuriEtAl2009a,FurerRaghavachari1994,KonemannRavi2003,KonemannRavi2002}.
Goemans~\cite{Goemans2006} gave a polynomial-time algorithm that finds a spanning tree of cost at most the optimum value that violates each degree bound by at most $2$.
Singh and Lau~\cite{SinghLau2007} improved the additive approximation guarantee to~$1$ by extending the iterative rounding method of Jain~\cite{Jain2001} with a relaxation step.
Zenklusen~\cite{Zenklusen2012} considered an extension of the problem where for every vertex $v$, the edges adjacent to $v$ have to be independent in a given matroid.
  
Motivated by a problem on binary matroids posed by Frienze, a matroidal generalization called the {\sc Minimum Bounded Degree Matroid Basis Problem} was introduced by Kir{\'a}ly, Lau and Singh~\cite{KiralyEtAl2012} in 2012.
The problem takes as input a matroid $M=(S,r)$, a cost function $c:S \rightarrow \mathbb{R}$, a hypergraph $H=(S, \mathcal{E})$ and lower and upper bounds $f,g:\mathcal{E}\rightarrow\mathbb{Z}_{\geq 0}$; the objective is to find a minimum-cost basis $B$ of $M$ such that $f(\eps) \leq |B \cap \eps| \leq g(\eps)$ for each $\eps \in \mathcal{E}$.
For this problem, the authors developed an approximation algorithm that is based on the iterative relaxation method and a clever token-counting argument of Chaudhuri et al.~\cite{ChaudhuriEtAl2009} and Singh and Lau~\cite{SinghLau2007}.
Let us denote the maximum degree of the hypergraph $H$ by $\Delta$.
When both lower bounds and upper bounds are present, their algorithm returns a basis $B$ of cost at most the optimum value such that $f(\eps) - 2\Delta + 1 \leq |B \cap \eps| \leq g(\eps) + 2\Delta - 1$ holds for each $\eps \in \mathcal{E}$.
Based on a technique of Bansal et al.~\cite{BansalEtAl2009}, they showed that the additive error can be improved when only lower bounds (or only upper bounds) are present, thus finding a basis of cost $B$ at most the optimum value such that $|B\cap \eps|\leq g(\eps)+\Delta-1$ (respectively, $f(\eps)-\Delta+1\leq |B\cap \eps|$) for each $\eps\in\mathcal{E}$.
Bansal et al.~\cite{BansalEtAl2013} considered extensions of the {\sc Minimum Bounded Degree Matroid Basis Problem} to contra-polymatroid intersection and to crossing lattice polyhedra.
In all of these cases, the solution for the problem is a $0{-}1$ vector defined on the ground set.

\subsection*{Our results}
In this paper we consider a different generalization of the {\sc Bounded Degree Matroid Basis Problem}.
The generalization deals with general polymatroids (or g-polymatroids) instead of matroids, and additionally allows multiplicities of the hyperedges.
Formally, the problem takes as input a g-polymatroid $Q(p,b)=(S,p,b)$ with a cost function $c:S \rightarrow \mathbb{R}$, and a hypergraph $H=(S, \mathcal{E})$ on the same ground set with lower and upper bounds $f, g:\mathcal{E}\rightarrow\mathbb{Z}_{\geq 0}$ and multiplicity vectors $m_\eps:S\rightarrow\mathbb{Z}_{>\geq0}$ for $\eps\in\Eps$ satisfying $m_\eps(s)=0$ for $s\in S-\eps$.
The objective is to find a minimum-cost element $x$ of~$Q(p,b)$ such that $f(\eps) \leq \sum_{s\in \eps}m_\eps(s)x(s) \leq g(\eps)$ for each $\eps \in \mathcal{E}$.
We call this problem the {\sc Bounded Degree g-polymatroid Element Problem with Multiplicities}. 
  
Our first main algorithmic result is the following:
\begin{theorem}
\label{thm:matroid1}
  There is a polynomial-time algorithm for the {\sc Bounded Degree g-polymatroid Element Problem with Multiplicities} which returns an element $x$ of $Q(p,b)$ of cost at most the optimum value such that $f(\eps)- 2\Delta+1 \leq \sum_{s\in \eps} m_\eps(s) x(s) \leq g(\eps)+2\Delta-1$ for each $\eps\in\Eps$, where $\Delta=\max_{s\in S}\left\{\sum_{\eps\in\Eps:s\in \eps} m_\eps(s)\right\}$.
\end{theorem}

Theorem~\ref{thm:matroid1} extends the result of Kir{\'a}ly et al.~\cite{KiralyEtAl2012} from matroids to g-poly\-matroids.
It turns out that, when upper bounds are present, there is a significant difference when g-polymatroids are considered instead of matroids.
Adapting the algorithm of Kir{\'a}ly et al. is not immediate, as a crucial step of their approach is to relax the problem by deleting a constraint corresponding to a hyperedge~$\eps$ with small $g(\eps)$ value.
This step is feasible when the solution is a $0$-$1$ vector, as in those cases the violation on $\eps$ is upper bounded by the size of the hyperedge.
This does not hold for g-polymatroids (or even for polymatroids), where an integral element might have coordinates larger than 1.
However, we show that after the first round of our algorithm, the problem can be restricted to the unit cube and so upper bounds remain tractable. 

When only lower bounds (or only upper bounds) are present, we call the problem {\sc Lower (Upper) Bounded Degree g-polymatroid Element Problem with Multiplicities}.
In this case, we show a similar result with an improved additive error:
\begin{theorem}
\label{thm:matroid2}
  There is an algorithm for the {\sc Lower Bounded Degree g-polymatroid Element Problem with Multiplicities} that runs in polynomial time and returns an element $x$ of $Q(p,b)$ of cost at most the optimum value such that $f(\eps)- \Delta+1 \leq \sum_{s\in \eps} m_\eps(s) x(s)$ for each $\eps\in\Eps$.
  An analogous result holds for the {\sc Upper Bounded Degree g-polymatroid Element Problem}, where $\sum_{s\in \eps} m_\eps(s) x(s) \leq g(\eps) + \Delta - 1$.
\end{theorem}

%\rnote{just commenting this out \\
%Theorem~\ref{thm:matroid1} extends the result of Kir{\'a}ly et al.~\cite{KiralyEtAl2012} to polymatroids when only lower bounds are present.
%In the light of this, one might be interested in obtaining similar results when only upper bounds (or both) are given.
%It turns out that there is significant difference when polymatroids are considered instead of matroids.
%Adapting the algorithm of Kir{\'a}ly et al.~\cite{KiralyEtAl2012} seems to be difficult, as a crucial step of their approach is to relax the problem by deleting a constraint corresponding to a hyperedge $\eps$ with small $g(\eps)$ value.
%This step is feasible when the solution is a $0{-}1$ vector, as in those cases the violation on $\eps$ is upper bounded by the size of the hyperedge.
%However, this does not hold for polymatroids where an integral element might have coordinates larger than $1$.}

While being interesting by itself, the algorithm alluded to in Theorem~\ref{thm:matroid2} serves as the key ingredient for our second main algorithmic result.
It concerns an extension of the {\sc Traveling Salesman Problem} (TSP), one of the cornerstones of combinatorial optimization.
In TSP, we are given a set of $n$ cities with their pairwise non-negative symmetric distances, and we seek a tour of minimum overall length that visits every city exactly once and returns to the origin.
For the metric variant, when distances obey the triangle inequality, Christofides~\cite{Christofides1976} in 1976 gave a polynomial-time algorithm that returns a 1.5-approximation to the optimal tour.
The algorithm was independently discovered by Serdyukov~\cite{Serdyukov1978}.
For more than 40 years, no polynomial-time algorithm with better approximation guarantee has been discovered.

In the generalization of the TSP, known as the {\sc Many-Visits TSP}, each city~$v$ is equipped with a request $r(v)\in\mathbb{Z}_{\geq 1}$, and we seek a tour of minimum overall length that visits city $v$ exactly $r(v)$~times and returns to the origin.
Note that a loop might have a positive cost at any city in this case. 
The {\sc Many-Visits TSP} was first considered in 1966 by Rothkopf~\cite{Rothkopf1966}.
The problem is clearly $\mathsf{NP}$-hard as it generalizes the TSP.
In 1980, Psaraftis \cite{Psaraftis1980} gave a dynamic programming algorithm with time complexity $\mathcal O(n^2\prod_{i=1}^n (r_i+1))$; observe that this value may be as large as $(r/n+ 1)^n$, which is prohibitive even for moderately large values of $r = \sum_{i=1}^n r_i$.
In 1984, Cosmadakis and Papadimitriou~\cite{CosmadakisPapadimitriou1984} designed a family of algorithms, the fastest of which has run time\footnote{The $\mathcal{O}^\star$ notation suppresses the factors polynomial in $n$.} $\mathcal O^\star(n^{2n}2^n + \log\sum r_i)$.

The analysis of the algorithm is highly non-trivial, combining graph-theoretic insights and involved estimates of various combinatorial quantities.
The usefulness of the Cosmadakis-Papadimitriou algorithm is limited by its superexponential dependence on~$n$ in the run time, as well as its superexponential space requirement.
Recently, Berger et al.~\cite{BergerEtAl2019} simultaneously improved the run time to $2^{\mathcal O(n)}\cdot \log \sum r_i$ and reduced the space complexity to polynomial.

As it is a generalization of the TSP, the {\sc Many-Visits TSP} is of fundamental interest.
This framework can be used for modeling \textit{high-multiplicity} scheduling problems~\cite{Psaraftis1980,HochbaumShamir1991,BraunerEtAl2005,vanderVeenZhang1996}.
In such problems, every job belongs to a job type, and two jobs of the same type are considered to be identical.
%  The number of job types is small (constant), while the total number of jobs is large.
One notable example of such problems is the \textit{aircraft sequencing problem}.
Airplanes are categorized into a small number of different classes.
Two airplanes belonging to the same class need the same amount of time to land.
In addition, there is a minimum time that should pass between the arrival of two planes.
The amount of this time only depends on the classes of the two airplanes, and the aim is to minimize the time when the last plane lands.

At the Hausdorff Workshop on Combinatorial Optimization in 2018, Rico Zenklusen brought up the topic of approximation algorithms for the metric version of {\sc Many-Visits TSP} in the context of iterative relaxation techniques; he suggested an approach to obtain a 1.5-approximation, which is unpublished.
%asked\footnote{The fourth author thanks Rico Zenklusen for posing the problem and initial discussions on the subject.} for a polynomial-time approximation algorithm for {\sc Many-Visits TSP} with metric cost functions.
The cost function being metric implies that the cost of each loop $c_{ii}$ is at most twice the cost of leaving city $i$ to any other city $j$ and returning.
The assumption of metric costs is necessary, as the TSP, and therefore the {\sc Many-Visits TSP} does not admit any non-trivial approximation for unrestricted cost functions.

Our next algorithmic result shows that 
%answers Zenklusen's question in a very strong form.
%Namely, we give 
a polynomial-time algorithm that matches the approximation guarantee of Christofides and Serdyukov for the single-visit case indeed exists.
\begin{theorem}
\label{thm:tsp1}
  There is a polynomial-time $1.5$-approximation for the metric {\sc Many-Visits TSP}.   
  %There is a polynomial-time algorithm that provides a $1.5$-approximation for the metric {\sc Many-Visits TSP}.
\end{theorem}

Let us remark that the requirements $r(v)$ are encoded in binary.
The TSP can also be formulated for directed graphs, where the cost function is asymmetric.
In a recent breakthrough, Svensson et al.~\cite{SvenssonEtAl2018} gave the first constant-factor approximation for the metric ATSP.
We can show the following:
\begin{theorem}
\label{thm:tsp2}
  There is a polynomial-time $\mathcal O(1)$-approximation for the metric {\sc Many-Visits ATSP}.
\end{theorem}
   
The rest of the paper is organized as follows.
In Sect.~\ref{sec:pre}, we give an overview of the notation and definitions.
In Sect.~\ref{sec:simple52approximation}, we provide a simple 2.5-approximation for the metric {\sc Many-Visits TSP} that runs in polynomial-time, and a polynomial-time constant-factor approximation for the metric {\sc Many-Visits ATSP}.
Thereafter, in Sect.~\ref{sec:bp}, we give the necessary background on g-polymatroids.
Sect.~\ref{sec:approxpolymatroid} describes the approximation algorithm for the {\sc Bounded Degree g-polymatroid Element Problem with Multiplicities}.
The 1.5-approximation for the metric {\sc Many-Visits TSP} is given in Sect.~\ref{sec:approx}.
We conclude in Sect.~\ref{sec:discussion}.

\section{Preliminaries}
\label{sec:pre}
Throughout the paper, we let $G=(V,E)$ be a finite, undirected complete graph on $n$ vertices, whose edge set $E$ also contains a self-loop at every vertex $v\in V$.
For a subset $F\subseteq E$ of edges, the \emph{set of vertices covered by $F$} is denoted by $V(F)$.
The \emph{number of connected components} of the graph $(V(F),F)$ is denoted by $\comp(F)$.
For a subset $X\subseteq V$ of vertices, the \emph{set of edges spanned by~$X$} is denoted by $E(X)$.
The set of edges incident to a vertex $v$ is denoted by~$\delta(v)$.
For a vector $x\in\mathbb{R}^{|E|}$, we denote the sum of the $x$-values on the edges incident to~$v$ by~$d_x(v)$.
Note that the $x$-value of the self-loop at $v$ is counted twice in~$d_x(v)$.
Given two graphs~$H_1,H_2$ on the same vertex set, $H_1+H_2$ denotes the multigraph on the same vertex set obtained by taking the union of the edge sets of~$H_1$ and~$H_2$.

Given a vector $x\in\mathbb{R}^{|S|}$ and a set $Z\subseteq S$, we use $x(Z)=\sum_{s\in Z} x(s)$.
The \emph{lower integer part of $x$} is denoted by $\floor{x}$, so $\floor{x} (s)=\floor{x(s)}$ for every $s\in S$.
This notation extends to sets as well, therefore by $\floor{x}(Z) $ we mean $\sum_{s \in Z} \floor{x}(s)$.
The \emph{support of $x$} is denoted by $\supp(x)$, that is, $\supp(x)=\{s\in S:x(s)\neq 0\}$.
The \emph{difference of set $B$ from set $A$} is denoted by $A-B=\{s\in A : s\notin B\}$.
We denote a single-element set $\{s\}$ by $s$, and with a slight abuse of notation, will write $A-s$ to indicate $A- \{s\}$.
The \emph{charasteristic vector} of a set $A$ is denoted by $\chi_A$.

Let $\varT$ be a collection of subsets of $S$.
We call $\mathcal{L} \subseteq \varT$ an \emph{independent laminar system} if for any pair $X, Y \in \mathcal{L}$: (i) they do not intersect, i.e. either $X \subseteq Y$, $Y \subseteq X$ or $X \cap Y = \emptyset$, (ii) the characteristic vectors $\chi_Z$ of the sets $Z \in \mathcal{L}$ are independent.
A \emph{maximal} independent laminar system $\mathcal{L}$ with respect to $\varT$ is an independent laminar system in $\varT$, such that for any $Y \in \varT-\mathcal{L}$ the system $\mathcal{L} \cup \{Y\}$ is not independent laminar.
In other words, if we include any set $Y$ from $\varT-\mathcal{L}$, it will intersect at least one set $Y$ from $\mathcal{L}$, or $\chi_Y$ can be given as a linear combination of $\{ \chi_Z: Z \in \mathcal{L} \}$. 
Given a laminar system~$\mathcal{L}$ and a set $X\subseteq S$, the set of maximal members of $\mathcal{L}$ lying inside $X$ is denoted by $\mathcal{L}^{\max}(X)$, that is, $\mathcal{L}^{\max}(X)=\{Y\in\mathcal{L}:\ Y\subset X,\ \not\exists Y'\in\mathcal{L}\ \text{s.t.}\ Y\subset Y'\subset X\}$.

The cost functions $c:E\rightarrow\mathbb{R}_{\geq 0}$ are assumed to satisfy the triangle inequality.
The \emph{minimum cost of an edge incident to a vertex $v$} is denoted by $c_v^{\min} := \min_{u \in V} c(uv)$.
Note that $u=v$ is allowed in the definition, therefore the minimum takes into account the cost of the self-loop at~$v$ as well. 
The triangle inequality holds for self-loops, too, meaning that $c(vv) \leq 2 \cdot c_v^{\min}$ for all $v\in V$.
%\begin{equation}
%\label{eq:tri_loop}
%  c(vv) \leq 2 \cdot c_v^{\min} \enspace .
%\end{equation}

In the {\sc Many-Visits TSP}, each vertex $v\in V$ is additionally equipped with a request $r(v)\in\mathbb {Z}_{\geq 1}$ encoded in binary.
The goal is to find a minimum-cost closed walk (or \emph{tour}) on the edges of the graph that visits each vertex $v\in V$ exactly $r(v)$ times.
Listing all the edges of such a walk might be exponential in the size of the input, hence we always consider \emph{compact representations} of the solution and the multigraphs that arise in our algorithms.
That is, rather than storing an $r(V)$-long sequence of edges, for every edge $e$ we store its multiplicity $z(e)$ in the solution.
As there are at most~$n^2$ different edges in the solution each having multiplicity at most $\max_{v \in V} r(v)$, the space needed to store a feasible solution is $\mathcal O(n^2\log r(V))$.
Therefore a vector $z \in \mathbb{Z}_{\geq 0}^{E}$ represents a feasible tour if $d_z(v)=2\cdot r(v)$ for every $v\in V$ and $\supp(z)$ is a connected subgraph of $G$.

From this compact representation, one can compute a collection $\mathcal{C}$ of pairs $(C, \mu_C)$, where each~$C$ is a simple closed walk (cycle) and $\mu_C$ is the corresponding integer denoting the number of copies of $C$.
The number of such cycles $C$ is polynomial in $n$, and one can compute $\mathcal{C}$ in polynomial time (see, e.g., the procedure in Sect.~2 of Grigoriev and van de Klundert~\cite{Grigoriev2006}).
One can obtain the explicit order of the vertices from $(C, \mu_C)$ the following way: traverse $\mu_C$ copies of an arbitrary cycle~$C$, and whenever a vertex $u$ is reached for the first time, traverse~$\mu_{C'}$ copies of every cycle $C' \neq C$ containing $u$.
Note that while the size of~$\mathcal{C}$ is polynomial in $n$, the size of the explicit order of the vertices is exponential, hence the time complexity of the last step is also exponential in $n$.

Denote by $\TR^\star_{c,r}$ an optimal solution for an instance $(G,c,r)$ of the {\sc Many-Visits TSP}, and by~$\TR^\star_{c,1}$ an optimal tour for the single-visit TSP (i.e., when $r(v)=1$ for each $v\in V$). 
Relaxing the connectivity requirement for solutions of the {\sc Many-Visits TSP} yields Hitchcock's transportation problem, which is solvable in polynomial time \cite{EdmondsKarp1970} and whose optimal solution we denote by~$\TP^\star_{c,r}$. 

\section{A Simple \texorpdfstring{$2.5$}{2.5}-Approximation for the Metric Many-Visits TSP}
\label{sec:simple52approximation}
In this section we give a simple $2.5$-approximation algorithm for the metric {\sc Many-Visits TSP}; see Algorithm~\ref{alg:apx_tp}.
\begin{algorithm}[h!]
  \caption{A polynomial-time $(\alpha+1)$-approximation for the metric {\sc Many-Visits TSP}.\label{alg:apx_tp}}
  \begin{algorithmic}[1]
    \Statex \textbf{Input:} A complete undirected graph $G$, costs $c:E\rightarrow\mathbb{R}_{\geq 0}$ satisfying the triangle inequality, requirements $r:V\rightarrow\mathbb{Z}_{\geq 1}$.
    \Statex \textbf{Output:} A tour that visits each $v \in V$ exactly $r(v)$ times. 
    \State Calculate an $\alpha$-approximate solution $\TR^\alpha_{c,1}$ for the single-visit metric TSP instance $(G,c,1)$. \label{st:i}      
    \State Calculate an optimal solution $\TP^\star_{c,r-1}$ for the transportation problem with prescriptions $r(v)-1$ for $v\in V$. \label{st:ii}
    \State \textbf{return} $T = \TR^\alpha_{c,1} + \TP^\star_{c,r-1}$ \label{st:iii}
  \end{algorithmic}
\end{algorithm}

\begin{theorem}
\label{thm:simple}
  The multigraph $T$ returned by Algorithm~\ref{alg:apx_tp} is a feasible solution to the metric {\sc Many-Visits TSP} instance $(G,c,r)$.
  The cost of the tour $T$ is at most $(\alpha+1)\cdot c(\TR^\star_{c,r})$.
\end{theorem}

\begin{proof}
  The degree of each vertex $v\in V$ is $2$ in $\TR^\alpha_{c,1}$, and is $2\cdot(r(v)-1)$ in~$\TP^\star_{c,r-1}$; hence the total degree of $v$ in $T = \TR^\alpha_{c,1} + \TP^\star_{c,r-1}$ is $2\cdot r(v)$, as required.
  Since~$\TR^\alpha_{c,1}$ is connected, $T = \TR^\alpha_{c,1} + \TP^\star_{c,r-1}$ is also connected, implying that it is a feasible solution to the problem.

  The cost of the tour $T$ constructed by Algorithm~\ref{alg:apx_tp} is equal to $c(T) = c(\TR^\alpha_{c,1}) + c(\TP^\star_{c,r-1})$.
  The cost of $\TR^\alpha_{c,1}$ is at most $\alpha\cdot c(\TR^\star_{c,1})$.
  Note that $c(\TR^\star_{c,1})\leq c(\TR^\star_{c,r})$, as the cost function satisfies the triangle inequality.
  %The cost of $\TP^\star_{c,r-1}$ is $\textsc{Opt}^{\TP}_{r-1}$.
  Again, by the triangle inequality, $c(\TP^\star_{c,r-1})\leq c(\TP^\star_{c,r})$.
  Hence we get
  \begin{align*}
    c(T) &= c(\TR^\alpha_{c,1}) + c(\TP^\star_{c,r-1})\\
         & \leq \alpha\cdot c(\TR^\star_{c,1}) + c(\TP^\star_{c,r-1})\\
                                                    & \leq \alpha\cdot c(\TR^\star_{c,r}) + c(\TP^\star_{c,r}) \\
                                                    &\leq(\alpha+1)\cdot c(\TR^\star_{c,r}),
  \end{align*}
  proving the approximation guarantee stated in the theorem.
\end{proof}

Christofides' algorithm~\cite{Christofides1976} for the single-visit metric TSP provides an approximate solution with $\alpha=1.5$; thus we get the following:
\begin{corollary}
  There is a polynomial-time algorithm that provides a $2.5$-approxi\-ma\-tion for the metric {\sc Many-Visits TSP}.
\end{corollary}
\begin{proof}
  The approximation ratio follows immediately; it remains to argue that the algorithm runs in polynomial time.

  Finding an approximate solution for the single-visit TSP in Step~\ref{st:i} requires $\mathcal O(n^3)$ operations~\cite{Christofides1976}.
  The transportation problem in Step~\ref{st:ii} can be solved in $\mathcal O(n^3\log r(V))$ operations using the Edmonds-Karp scaling method~\cite{EdmondsKarp1970}.
  Finally, Step~\ref{st:iii} takes $\mathcal O(n^2 \log r(V))$ operations, therefore the total time complexity of the algorithm is $\mathcal O(n^3 \log r(V))$.
\end{proof}

For the metric {\sc Many-Visits ATSP}, in Step~\ref{st:i} of Algorithm~\ref{alg:apx_tp} we can apply the $\mathcal O(1)$-approximation for metric ATSP due to Svensson et al.~\cite{SvenssonEtAl2018}.
This leads to the proof of Theorem~\ref{thm:tsp2}.

\section{Polyhedral background}
\label{sec:bp}
In what follows, we make use of some basic notions and theorems of the theory of generalized polymatroids.
For background, see for example the paper of Frank and Tardos~\cite{frank1988generalized} or Chapter~14 in the book by Frank~\cite{frank2012connections}.

Given a ground set $S$, a set function $b:2^S\rightarrow\mathbb{Z}$ is \emph{submodular} if
\begin{equation*}
  b(X)+b(Y)\geq b(X\cap Y) + b(X\cup Y)
\end{equation*} 
holds for every pair of subsets $X,Y\subseteq S$.
A set function $p:2^S\rightarrow\mathbb{Z}$ is \emph{supermodular} if $-p$ is submodular.
As a generalization of matroid rank functions, Edmonds introduced the notion of polymatroids~\cite{Edmonds1970}.
A set function $b$ is a \emph{polymatroid function} if $b(\emptyset)=0$, $b$ is non-decreasing, and $b$ is submodular.

We define
\begin{equation*}
  P(b):=\{x\in\mathbb{R}^{S}_{\geq 0}: x(Y)\leq b(Y)\ \text{for every}\ Y\subseteq S\} \enspace.
\end{equation*}
The set of integral elements of $P(b)$ is called a \emph{polymatroidal set}.
Similarly, the \emph{base polymatroid}~$B(b)$ is defined by
\begin{equation*}
  B(b):=\{x\in\mathbb{R}^{S}: x(Y)\leq b(Y)\ \text{for every}\ Y\subseteq S, \, x(S)=b(S)\} \enspace .
\end{equation*}
Note that a base polymatroid is just a facet of the polymatroid $P(b)$.
In both cases, $b$ is called the \emph{border function} of the polyhedron.
Although non-negativity of $x$ is not assumed in the definition of~$B(b)$, this follows by the monotonicity of $b$ and the definition of $B(b)$: $x(s)=x(S)-x(S-s) \geq b(S)-b(S-s)\geq 0$ holds for every $s\in S$.
The set of integral elements of~$B(b)$ is called a \emph{base polymatroidal set}.
Edmonds~\cite{Edmonds1970} showed that the vertices of a polymatroid or a base polymatroid are integral, thus $P(b)$ is the convex hull of the corresponding polymatroidal set, while $B(b)$ is the convex hull of the corresponding base polymatroidal set.
For this reason, we will call the sets of integral elements of $P(b)$ and $B(b)$ simply a polymatroid and a base polymatroid.

Hassin~\cite{hassin1982minimum} introduced polyhedra bounded simultaneously by a nonnegative, monotone non-decreasing submodular function $b$ over a ground set $S$ from above and by a nonnegative, monotone non-decreasing supermodular function $p$ over $S$ from below, satisfying the so-called \emph{cross-inequality} linking the two functions:
\begin{equation*}
  b(X) - p(Y) \geq b(X - Y) - p(Y - X)\qquad~\mbox{for every pair of subsets}~X,Y\subseteq S \enspace .
\end{equation*}
We say that a pair $(p,b)$ of set functions over the same ground set $S$ is a \emph{paramodular pair} if $p(\emptyset)=b(\emptyset)=0$, $p$ is supermodular, $b$ is submodular, and the cross-inequality holds.
The slightly more general concept of generalized polymatroids was introduced by Frank~\cite{frank1984generalized}.
A \emph{generalized polymatroid}, or \emph{g-polymatroid} is a polyhedron of the form
\begin{equation*}
  Q(p,b):=\left\{ x \in \mathbb{R}^{S}: p(Y) \leq x(Y) \leq b(Y) \ \text{for every}\ Y \subseteq S \right\} \enspace , 
\end{equation*}
where $(p,b)$ is a paramodular pair.
Here, $(p,b)$ is called the \emph{border pair} of the polyhedron. 
It is known (see e.g. \cite{frank2012connections}) that a g-polymatroid defined by an integral paramodular pair is a non-empty integral polyhedron.

A special g-polymatroid is a box $T(\ell,u)=\{x\in \mathbb{R} \sp {S}:  \ell\leq x\leq u\}$ where $\ell:S\rightarrow \mathbb{Z} \cup \{-\infty \}$, $u:S\rightarrow \mathbb{Z}\cup \{\infty \}$ with $\ell\leq u$. 
Another illustrious example is base polymatroids. 
Indeed, given a polymatroid function $b$ with finite $b(S)$, its \emph{complementary set function} $p$ is defined for $X\subseteq S$ by $p(X):=b(S)-b(S-X)$. 
It is not difficult to check that $(p,b)$ is a paramodular pair and that $B(b)=Q(p,b)$.

The intersection $Q'$ of a g-polymatroid $Q=Q(p,b)$ and a box $T=T(\ell,u)$ is non-empty if and only if $\ell(Y)\leq b(Y)$ and $p(Y)\leq u(Y)$ hold for every $Y\subseteq S$. 
When $Q'$ is non-empty, its unique border pair $(p',b')$ is given by 
\begin{equation*}
  p'(Z) = \max\{ p(Z') - u(Z'-Z)+ \ell(Z-Z') : Z'\subseteq S\},
\end{equation*} 
\begin{equation*}
  b'(Z) = \min\{ b(Z') - \ell(Z'-Z)+ u(Z-Z') :  Z'\subseteq S\} \enspace .
\end{equation*} 

Given a g-polymatroid $Q(p, b)$ and $Z\subset S$, by \emph{deleting} $Z\subseteq S$ from $Q(p,b)$ we obtain a g-polymatroid $Q(p, b)\setminus Z$ defined on set $S - Z$ by the restrictions of~$p$ and $b$ to $S - Z$, that is,
\begin{equation*}
  Q(p, b)\setminus Z:=\{x\in\mathbb{R}^{S-Z}: p(Y) \leq x(Y)\leq b(Y)\ \text{for every}\ Y\subseteq S - Z\} \enspace .
\end{equation*}
In other words, $Q(p, b)\setminus Z$ is the projection of $Q(p, b)$ to the coordinates in $S - Z$.

Extending the notion of contraction is not immediate.
A set can be naturally identified with its characteristic vector, that is, contraction is basically an operation defined on $0{-}1$ vectors.
In our proof, we will need a generalization of this to the integral elements of a g-polymatroid.
However, such an element might have coordinates larger than one as well, hence finding the right definition is not straightforward. 
In the case of matroids, the most important property of contraction is the following: $I$ is an independent of $M/Z$ if and only if $F\cup I$ is independent in~$M$ for any maximal independent set~$F$ of~$Z$.
  
With this property in mind, we define the g-polymatroid obtained by the contraction of an integral vector $z\in Q(p,b)$ to be the polymatroid $Q(p',b'):=Q(p,b)/z$ on the same ground set~$S$ with the border functions
\begin{align*}
  p'(X) &:= p(X) - z(X) \\
  b'(X) &:= b(X) - z(X) \enspace .
\end{align*}
Observe that $p'$ is obtained as the difference of a supermodular and a modular function, implying that it is supermodular.
Similarly, $b'$ is submodular.
Moreover, $p'(\emptyset)=b'(\emptyset)=0$, and
\begin{align*}
  b'(X)-p'(Y)
  {}&{}=
  b(X)-z(X)-p(Y)+z(Y)\\
  {}&{}\geq
  b(X-Y)+p(Y-X)-z(X-Y)+z(Y-X)\\
  {}&{}=
  b'(X-Y)-p'(Y-X),
\end{align*}
hence $(p',b')$ is indeed a paramodular pair.
The main reason for defining the contraction of an element $z\in Q(p,b)$ is shown by the following lemma.
\begin{lemma}
\label{lem:contraction}
  Let $Q(p',b')$ be the polymatroid obtained by contracting $z\in Q(p,b)$.
  Then $x+z\in Q(p,b)$ for every $x\in Q(p',b')$.
\end{lemma}

\begin{proof}
  Let $x\in Q(p',b')$. By definition, this implies $p'(Y)\leq x(Y)\leq b'(Y)$ for $Y\subseteq S$. 
  Thus $p(Y)=p'(Y)+z(Y)\leq x(Y)+z(Y)\leq b'(Y)+z(Y)=b(Y)$, concluding the proof.
\end{proof}

Formally, the {\sc Bounded Degree g-po\-ly\-mat\-ro\-id Element Problem} takes as input a g-polymatroid $Q(p,b)$ with a cost function $c:S \rightarrow \mathbb{R}$, and a hypergraph $H=(S, \Eps)$ on the same ground set with lower and upper bounds $f,g:\Eps\rightarrow\mathbb{Z}_{\geq 0}$ and multiplicity vectors $ m_\eps:S\rightarrow\mathbb{Z}_{\geq0}$ for $\eps\in\Eps$ satisfying $m_\eps(s)=0$ for $s\in S-\eps$.
The objective is to find a minimum-cost element $x$ of $Q(p,b)$ such that $f(\eps) \leq \sum_{s\in \eps} m_\eps(s) x(s) \leq g(\eps)$ for each $\eps \in \Eps$.
  
\section{Approximating the Bounded Degree g-polymatroid Element Problem with Multiplicities}
\label{sec:approxpolymatroid}
The aim of this section is to prove Theorems~\ref{thm:matroid1} and~\ref{thm:matroid2}.
We start by formulating a linear programming relaxation for the {\sc Bounded Degree g-polymatroid Element Problem}:

\leqnomode

\begin{align*}
  \label{eq:lp_poly}
  \text{minimize} \qquad \sum_{s \in S}  c(s) \ &x(s) \\
  \text{subject to} \qquad p(Z)  \leq \ &x(Z) \leq b(Z) &\forall Z \subseteq S \tag{LP}\\
  \qquad f(\eps) \leq \sum_{s\in \eps} \ m_\eps(s) \, &x(s) \leq g(\eps) &\forall \eps \in \Eps %\\
%  \qquad &x(s) \geq 0 &\forall s\in S
\end{align*}

Although the program has an exponential number of constraints, it can be separated in polynomial time using submodular minimization \cite{iwata2001,mccormick2005,schrijver2000}.
Algorithm~\ref{alg:matd} generalizes the approach by Kir{\'a}ly et al.~\cite{KiralyEtAl2012}.
We iteratively solve the linear program, delete elements which get a zero value in the solution, update the solution values and perform a contraction on the polymatroid, or remove constraints arising from the hypergraph.
There is a significant difference between the first round of the algorithm and the later ones. 
In the first round, the bounds on the coordinates solely depend on $p$ and $b$, while in the subsequent rounds the whole problem is restricted to the unit cube.
It is somewhat surprising that this restriction affects neither the solvability of the problem nor the additive error.
Intuitively, the very first step of the algorithm fixes `most part' of each coordinate, and the following steps are changing their value by at most $1$.

\clearpage
\pagebreak

\begin{algorithm}[h!]
\caption{Approximation algorithm for the {\sc Bounded Degree g-poly\-matroid Element Problem with Multiplicities}}\label{alg:matd}
\begin{algorithmic}[1]
  \Statex \textbf{Input:} A g-polymatroid $Q(p,b)$ on ground set $S$, cost function $c:S\rightarrow\mathbb{R}$, a hypergraph $H = (S, \Eps)$, lower bounds $f,g:\Eps\rightarrow\mathbb{Z}_{\geq 0}$, multiplicities $m_\eps:S\rightarrow\mathbb{Z}_{\geq 0}$ for $\eps \in \Eps$ satisfying $m_\eps(s)=0$ for $s\in S-\eps$.
  \Statex \textbf{Output:} $z\in Q(p,b)$ of cost at most $\textsc{OPT}_{LP}$, violating the hyperedge constraints by at most $2\Delta-1$. 
  \State Initialize $z(s) \leftarrow 0$ for every $s\in S$.
  \While{$S\neq\emptyset$}
    \State \begin{varwidth}[t]{0.9\linewidth} Compute a basic optimal solution $x$ for \eqref{eq:lp_poly}. \\ (Note: starting from the second iteration, $0\leq x \leq 1$.) \end{varwidth}
    \begin{algsubstates}\
      \State \begin{varwidth}[t]{0.9\linewidth}
      Delete any element $s$ with $x(s)=0$.
      Update each hyperedge $\eps\leftarrow \eps-s$ and $m_\eps(s)\leftarrow 0$.
      Update the g-polymatroid $Q(p,b)\leftarrow Q(p,b)\setminus s$ by deletion. \label{st:del}
     	\end{varwidth}
     	\State \begin{varwidth}[t]{0.9\linewidth}
      For all $s\in S$ update $z(s) \leftarrow z(s) + \floor{x}(s)$.\\
      Apply polymatroid contraction $Q(p,b)\leftarrow Q(p,b)/\floor{x}$, that is, redefine $p(Y) := p(Y) - \floor{x}(Y)$ and $b(Y) := b(Y) - \floor{x}(Y)$ for every $Y \subseteq S$.\\
      Update $f(\eps) \leftarrow f(\eps) - \displaystyle\sum_{s\in \eps}\, m_\eps(s) \floor{x}(s)$ and $g(\eps) \leftarrow g(\eps) - \displaystyle\sum_{s\in \eps}\, m_\eps(s) \floor{x}(s)$ for each $\eps \in \Eps$.\label{st:inc}
      \end{varwidth}
    	\State If $m_\eps(\eps) \leq 2\Delta-1$, let $\Eps \leftarrow \Eps - \eps$. \label{st:rem}
      \State \begin{varwidth}[t]{0.9\linewidth} \textbf{if} it is the first iteration \textbf{then} \label{st:first}\\
      Take the intersection of $Q(p,b)$ and the unit cube $[0,1]^S$, that is, $p(Y):=\max\{ p(Y') - |Y'-Y| : Y'\subseteq S\}$ and $b(Y) := \min\{ b(Y')+ |Y-Y'| :  Y'\subseteq S\}$ for every $Y\subseteq S$.
      \end{varwidth}    	    	 
      \end{algsubstates}
    \EndWhile
  \State \textbf{return} $z$
\end{algorithmic}
\end{algorithm}

\begin{proof}[Proof of Theorem~\ref{thm:matroid1}]
$\;$
   \paragraph{\textbf{Correctness}}
  First we show that if the algorithm terminates then the returned solution $z$ satisfies the requirements of the theorem.  
  In a single iteration, the g-polymatroid $Q(p,b)$ is updated to $(Q(p,b)\setminus D)/\floor{x}$, where \mbox{$D=\{s:x(s)=0\}$} is the set of deleted elements.
  In the first iteration, the g-polymatroid thus obtained is further intersected with the unit cube.
  By Lemma~\ref{lem:contraction}, the vector $x-\lfloor x\rfloor$ restricted to $S-D$ remains a feasible solution for the modified linear program in the next iteration.
  Note that this vector is contained in the unit cube as its coordinates are between $0$ and $1$.
  This remains true when a lower degree constraint is removed in Step~\ref{st:rem} as well, therefore the cost of $z$ plus the cost of an optimal LP solution does not increase throughout the procedure.
  Hence the cost of the output~$z$ is at most the cost of the initial LP solution, which is at most the optimum.
  
  By Lemma~\ref{lem:contraction}, the vector $x-\lfloor x \rfloor+z$ is contained in the original g-polymatroid, although it might violate some of the lower and upper bounds on the hyperdeges.
  We only remove the constraints corresponding to the lower and upper bounds for a hyperedge $\eps$ when $m_\eps(\eps) \leq 2\Delta-1$.
  As the g-polymatroid is restricted to the unit cube after the first iteration, these constraints are violated by at most $2\Delta-1$, as the total value of $\sum_{s\in\eps}m_\eps(s)z(s)$ can change by a value between~$0$ and $2\Delta-1$ in the remaining iterations.
 
  It remains to show that the algorithm terminates successfully.
  The proof is based on similar arguments as in Kir{\'a}ly et al.~\cite[proof of Theorem 2]{KiralyEtAl2012}.

  \paragraph{\textbf{Termination}}
  Suppose, for sake of contradiction, that the algorithm does not terminate.
  Then there is some iteration after which none of the simplifications in Steps~\ref{st:del}-\ref{st:rem} can be performed.
  This implies that for the current basic LP solution $x$ it holds $0<x(s)<1$ for each $s \in S$ and $m_\eps(\eps)\geq 2\Delta$ for each $\eps \in \Eps$.
  We say that a set $Y$ is \emph{p-tight} (or \emph{b-tight}) if $x(Y)=p(Y)$ (or $x(Y)=b(Y)$), and let $\varT^p=\{Y\subseteq S:x(Y)=p(Y)\}$ and $\varT^b=\{Y\subseteq S:x(Y)=b(Y)\}$ denote the collections of $p$-tight and $b$-tight sets with respect to solution $x$.
  
  Let $\varL$ be a maximal independent laminar system in $\varT^p \cup \varT^b$.
  \begin{claim}
  \label{claim:uncrossing}
    $\spa{(\{\chi_Z: Z \in \varL\})} = \spa{(\{\chi_Z: Z \in \varT^p \cup \varT^b\})}$
  \end{claim}
  \begin{proof}[Proof of Claim~\ref{claim:uncrossing}.]
  \renewcommand{\qedsymbol}{$\Diamond$}
    The proof uses an uncrossing argument.
    Let us suppose indirectly that there is a set $R$ from~$\varT^p \cup \varT^b$ for which $\chi_R \notin \spa{(\{\chi_Z: Z \in \varL\})}$.
    Choose this set $R$ so that it is incomparable to as few sets of $\varL$ as possible.
    Without loss of generality, we may assume that $R \in \varT^p$.
    Now choose a set $T \in \varL$ that is incomparable to $R$.
    Note that such a set necessarily exists as the laminar system is maximal.
    We distinguish two cases.
  
    \noindent \textbf{Case 1.} $T \in \varT^p$.
    Because of the supermodularity of $p$, we have
    \begin{align*}
      x(R) + x(T) &= p(R) + p(T) \leq p(R \cup T) + p(R \cap T) \leq x(R \cup T) + x(R \cap T)\\
                  &= x(R) + x(T) \enspace ,  
    \end{align*}
    hence equality holds throughout.
    That is, $R \cup T$ and $R \cap T$ are in $\varT^p$ as well.
    In addition, since $\chi_R + \chi_T = \chi_{R \cup T} + \chi_{R \cap T}$ and $\chi_R$ is not in $\spa{(\{\chi_Z: Z \in \varL\})}$, either $\chi_{R \cup T}$ or $\chi_{R \cap T}$ is not contained in $\spa{(\{\chi_Z: Z \in \varL\})}$.
    However, both $R \cup T$ and $R \cap T$ are incomparable with fewer sets of $\varL$ than~$R$, which is a contradiction.
  
    \noindent \textbf{Case 2.} $T \in \varT^b$.
    Because of the cross-inequality, we have
    \begin{align*}
      x(T) - x(R) &= b(T) - p(R) \geq b(T \setminus R) - p(R \setminus T) \geq x(T \setminus R) - x(R \setminus T)\\
                  &= x(T) - x(R) \enspace ,
    \end{align*}
    implying $T \setminus R \in \varT^b$ and $R \setminus T \in \varT^p$.
    Since $\chi_R + \chi_T = \chi_{R \setminus T} + \chi_{R \setminus T} + 2 \ \chi_{R \cup T}$ and $\chi_R$ is not in $\spa{(\{\chi_Z: Z \in \varL\})}$, one of the vectors $\chi_{R \setminus T}$, $\chi_{R \setminus T}$ and $\chi_{R \cup T}$ is not contained in $\spa{(\{\chi_Z: Z \in \varL\})}$.
    However, any of these three sets is incomparable with fewer sets of $\varL$ than~$R$, which is a contradiction.
  
    The case when $R \in \varT^b$ is analogous to the above.
    This completes the proof of the Claim.
  \end{proof}
  
  We say that a hyperedge $\eps \in \Eps$ is \emph{tight} if $f(\eps)=\sum_{s\in \eps} m_\eps(s) x(s)$ or $g(\eps)=\sum_{s\in \eps} m_\eps(s) x(s)$. % m_\eps(s)
  As $x$ is a basic solution, there is a set $\Eps'\subseteq\Eps$ of tight hyperedges such that $\{m_\eps: \eps \in \Eps'\}\cup\linebreak \{\chi_Z: Z \in \varL\}$ are linearly independent vectors with $|\Eps'|+|\varL|=|S|$.
  
  We derive a contradiction using a token-counting argument.
  We assign $2\Delta$ tokens to each element $s \in S$, accounting for a total of $2\Delta |S| $ tokens.
  The tokens are then redistributed in such a way that each hyperedge in $\Eps'$ and each set in $\varL$ collects at least $2\Delta$ tokens, while at least one extra token remains.
  This implies that $2\Delta |S|>2\Delta|\Eps'|+2\Delta|\varL|$, leading to a contradiction.
  
  We redistribute the tokens as follows.
  Each element $s$ gives $\Delta$ tokens to the smallest member in~$\varL$ it is contained in, and $m_\eps(s)$ token to each hyperedge $\eps\in\Eps'$ it is contained in.
  As \mbox{$\sum_{\eps\in\Eps:s\in \eps} m_\eps(s)\leq\Delta$} holds for every element $s \in S$, thus we redistribute at most $2\Delta$ tokens per element and so the redistribution step is valid.  
  Now consider any set $U\in\varL$. 
  Recall that $\varL^{\max}(U)$ consists of the maximal members of $\varL$ lying inside $U$. 
  Then $U-\bigcup_{W\in\varL^{\max}(U)} W\neq\emptyset$, as otherwise $\chi_U=\sum_{W\in\varL^{\max}(U)} \chi_W$, contradicting the independence of $\varL$.
  For every set $Z$ in~$\varL$, $x(Z)$ is an integer, meaning that $x(U - \bigcup_{W\in\varL^{\max}(U)} W)$ is an integer.  
  But also $0 < x(s) < 1$ for every $s \in S$, which means that $U - \bigcup_{W\in\varL^{\max}(U)} W$ contains at least 2 elements.
  Therefore, each set $U$ in $\varL$ receives at least $2 \Delta$ tokens, as required.
  By assumption, $m_\eps(\eps) \geq 2 \Delta$ for every hyperedge $\eps \in \Eps'$, which means that each hyperedge in $\Eps'$ receives at least $2\Delta$~tokens, as required. 

  If $\sum_{\eps\in\Eps':s\in \eps} m_\eps(s)\leq \Delta$ holds for any $s\in S$ or $\mathcal{L}^{\max}(S)$ is not a partition of $S$, then an extra token exists.
  Otherwise, $\sum_{\eps\in\Eps'}m_{\eps}=\Delta\cdot\chi_S=\sum_{W\in\varL^{\max}(S)}^q\chi_W$, contradicting the independence of $\{m_\eps: \eps \in \Eps'\}\cup \{\chi_Z: Z \in \varL\}$.
     
  \paragraph{\textbf{Time complexity}}
  Let us now prove that the run time of the algorithm is polynomial in the input size.
  Solving an LP, as well as removing an element from a hyperedge in Step~\ref{st:rem} or removing a hyperedge in Step~\ref{st:del} can be done in polynomial time.
  Now let us turn to the g-polymatroid contraction in Step~\ref{st:inc} and taking the intersection with the unit cube in Step~\ref{st:first}.
  The function value is not recalculated for every subset $Y \subseteq S$, as there is an exponential number of such subsets.
  Instead, we calculate the value of the current functions~$p$ and $b$ for a set $Y$ only when it is needed during the ellipsoid method.
  We keep track of the vectors $\floor{x}$ that arise during contraction steps (there is only a polynomial number of them), and every time a query for $p$ or $b$ happens, it takes into account every contraction and removal that occurred until that point.

  Let us now bound the number of iterations.
  In every iteration at least one of Steps~\ref{st:del}-\ref{st:rem} is executed.
  Clearly, Step~\ref{st:del} can be repeated at most $|S|$ times, while Step~\ref{st:rem} can be repeated at most $|\Eps|$ times.
  Starting from the second iteration, we are working in the unit cube. 
  That is, when Step~\ref{st:inc} adds the integer part of a variable $x(s)$ to $z(s)$ and reduces the problem, then the given variable will be $0$ in the next iteration and so element $s$ is deleted. 
  This means that the total number of iterations of Step~\ref{st:inc} is at most $\mathcal O(|S|)$.
  
  We therefore showed that the number of iterations, as well as the time complexity of each step taken by the algorithm can be bounded by the input size, meaning the algorithm runs in polynomial time.
\end{proof}

Now we turn to the proof of the case when only lower or only upper bounds are given.
\begin{proof}[Proof of Theorem~\ref{thm:matroid2}]
  The proof is similar to the proof of Theorem~\ref{thm:matroid1}, the main difference appears in the the counting argument.
  When only lower bounds are present, the condition in Step~\ref{st:rem} changes: we delete a hyperedge $\eps$ if $f(\eps)\leq\Delta-1$.
  Suppose, for the sake of contradiction, that the algorithm does not terminate.
  Then there is an iteration after which none of the simplifications in Steps~\ref{st:del}-\ref{st:rem} can be performed.
  This implies that in the current basic solution $0 < x(s) < 1$ holds for each $s \in S$ and $f(\eps) \geq \Delta$ for each $\eps \in \Eps$.
  We choose a subset $\Eps'\subseteq\Eps$ and a maximal independent laminar system~$\varL$ of tight sets the same way as in the proof of Theorem~\ref{thm:matroid1}.
  Recall that $|\Eps'| + |\varL| = |S|$.
  
  Let $Z_1, \dots, Z_k$ denote the members of the laminar system $\varL$. 
  As $\varL$ is an independent system, $Z_i-\bigcup_{W\in\mathcal{L}^{\max}(Z_i)}W\neq\emptyset$.
  Since $x(s)<1$ for all $s\in S$, $x(Z_i-\bigcup_{W\in\mathcal{L}^{\max}(Z_i)}W)<|Z_i-\bigcup_{W\in\mathcal{L}^{\max}(Z_i)}W|$.
  As we have integers on both sides of this inequality, we get 
  \begin{equation*}
    |Z_i-\!\!\!\bigcup_{W\in\mathcal{L}^{\max}(Z_i)}\!\!\!\!\!\!W|-x(Z_i-\!\!\!\bigcup_{W\in\mathcal{L}^{\max}(Z_i)}\!\!\!\!\!\!W)\geq 1\quad\text{for all}\ i=1,\dots,k \enspace .
  \end{equation*}
  Moreover, $\sum_{s\in\eps}m_{\eps}(s)x(s)\geq f(\eps)\geq\Delta$ for all hyperedges; therefore,
  \begin{align*}
    |\Eps'| + |\varL|
    {}&{}\leq 
    \sum_{\eps \in \Eps'} \frac{\sum_{s \in \eps} m_\eps(s) x(s)}{\Delta} + \sum_{i=1}^k \left[ |Z_i - \!\!\! \bigcup_{W \in \mathcal{L}^{\max}(Z_i)} \!\!\!\!\!\! W| - x(Z_i - \!\!\! \bigcup_{W \in \mathcal{L}^{\max}(Z_i)} \!\!\!\!\!\! W) \right] \\
    {}&{}= 
    \sum_{s \in S} \frac{x(s)}{\Delta} \sum_{\substack{\eps \in \Eps' \\ s\in \eps}} m_\eps(s) + \sum_{W \in \mathcal{L}^{\max}(S)}|W| - \sum_{W \in \mathcal{L}^{\max}(S)} x(W) \label{eq:optional} \\
    {}&{}\leq 
    |S| \enspace .
  \end{align*}
  In the last line, the first term is at most $x(S)$ since $\sum_{\eps\in\Eps:s\in \eps} m_\eps(s)\leq\Delta$ holds for each element~\mbox{$s \in S$}.
  From $x(S)- \sum_{W \in \mathcal{L}^{\max}(S)} x(W)\leq |S|-\sum_{W \in \mathcal{L}^{\max}(S)}|W|$ the upper bound of $|S|$ follows.
  As $|S| = |\varL| + \mathcal{|\Eps'|}$, we have equality throughout. 
  This implies that $\sum_{\eps \in \Eps'} m_\eps = \Delta \cdot \chi_S=\Delta\cdot\sum_{W\in\mathcal{L}^{\max}(S)}\chi_W$, contradicting linear independence.
  
  If only upper bounds are present, we remove a hyperedge $\eps$ in Step~\ref{st:rem} when $g(\eps)+\Delta-1 \geq m_\eps(\eps)$. 
  Suppose, for the sake of contradiction, that the algorithm does not terminate.
  Then there is an iteration after which none of the simplifications in Steps~\ref{st:del}-\ref{st:rem} can be performed.
  This implies that in the current basic solution $0 < x(s) < 1$ holds for each $s \in S$ and $m_\eps(\eps)-g(\eps) \geq \Delta$ for each $\eps \in \Eps$.
  Again, we choose a subset $\Eps'\subseteq\Eps$ and a maximal independent laminar system~$\varL$ of tight sets the same way as in the proof of Theorem~\ref{thm:matroid1}. 
     
  Let $Z_1, \dots, Z_k$ denote the members of the laminar system $\varL$. 
  As $\varL$ is an independent system, $Z_i-\bigcup_{W\in\mathcal{L}^{\max}(Z_i)}W\neq\emptyset$ and so
  \begin{equation*}
    x(Z_i - \!\!\! \bigcup_{W \in \mathcal{L}^{\max}(Z_i)} \!\!\!\!\!\! W) \geq 1 \enspace .
  \end{equation*}

  By $\sum_{s \in \eps} m_\eps(s) x(s) \leq g(\eps)$, we get $\sum_{s \in \eps} m_\eps(s)-\sum_{s \in \eps} m_\eps(s) x(s) \geq m_\eps(\eps)-g(\eps) \geq \Delta$. 
  Therefore,
  \begin{align*}
    |\Eps'| + |\varL|
    {}&{}\leq
    \sum_{\eps \in \Eps'} \frac{\sum_{s \in \eps} m_\eps(s)-\sum_{s \in \eps} m_\eps(s) x(s)}{\Delta} + \sum_{i=1}^k x(Z_i - \!\!\! \bigcup_{W \in \mathcal{L}^{\max}(Z_i)} \!\!\!\!\!\! W) \\
    {}&{}=
    \sum_{s \in S} \frac{1-x(s)}{\Delta} \sum_{\substack{\eps \in \Eps' \\ s\in \eps}} m_\eps(s) + \sum_{W \in \mathcal{L}^{\max}(S)} x(W) \label{eq:optional2} \\
    {}&{}\leq
    \sum_{s \in S} \frac{1-x(s)}{\Delta} \sum_{\substack{\eps \in \Eps' \\ s\in \eps}} m_\eps(s) +  x(S) \leq |S| \enspace .
  \end{align*} 
  In the last line, the first term is at most $|S|-x(S)$ since $\sum_{\eps\in\Eps:s\in \eps} m_\eps(s)\leq\Delta$ holds for every element $s \in S$.
  Therefore, the upper bound of $|S|$ follows.
  As $|S| = |\varL| + \mathcal{|\Eps'|}$, we have equality throughout. 
  %This implies that $\sum m_\eps(s) = \Delta$ for $s\in S$ and necessarily $Z_k=S$.
  This implies that $\sum_{\eps \in \Eps'} m_\eps = \Delta \cdot \chi_S=\Delta\cdot\sum_{W\in\mathcal{L}^{\max}(S)}\chi_W$, contradicting linear independence.
\end{proof}

We have seen in Sect.~\ref{sec:bp} that base polymatroids are special cases of g-polymatroids.
This implies that the results of Theorem~\ref{thm:matroid2} immediately apply to polymatroids.
Let us first formally define the problem.

In the {\sc Lower Bounded Degree Polymatroid Basis Problem with Multiplicities}, we are given a base polymatroid $B(b)=(S,b)$ with a cost function $c:S \rightarrow \mathbb{R}$, and a hypergraph $H=(S, \Eps)$ on the same ground set.
The input contains lower bounds $f: \Eps \rightarrow \mathbb{Z}_{\geq 0}$ and multiplicity vectors $m_\eps: \eps \rightarrow \mathbb{Z}_{\geq 1}$ for every hyperedge $\eps \in \Eps$.
The objective is to find a minimum-cost element $x \in B(b)$ such that $f(\eps) \leq \sum_{s \in \eps} m_\eps(s) x(s)$ holds for each $\eps \in \Eps$.

\begin{corollary}
\label{thm:polym}
  There is an algorithm for the {\sc Lower Bounded Degree Polymatroid Basis Problem with Multiplicities} that runs in polynomial time and returns an element $x$ of $B(b)$ of cost at most the optimum value such that $f(\eps)- \Delta+1 \leq \sum_{s \in \eps} m_\eps(s) x(s)$ for each $\eps \in\Eps$.
\end{corollary}  

\clearpage
\pagebreak

\section{A \texorpdfstring{$1.5$}{1.5}-Approximation for the Metric Many-Visits TSP}
\label{sec:approx}
In this section we design a polynomial-time $1.5$-approximation for the {\sc Metric Many-Visits TSP}.
Our approach is along similar lines as Christofides' algorithm~\cite{Christofides1976} for the metric single-visit TSP.
It constructs a solution in three steps: (i) it computes a minimum cost spanning tree that ensures the connectivity of the solution, then (ii) it adds a minimum cost matching  on the set of vertices of odd degree in order to obtain an Eulerian subgraph, and finally (iii) it forms a Hamiltonian circuit from an Eulerian circuit by shortcutting repeated vertices.
 
In our setting of many-visits, we make use of the following formulation of the metric {\sc Many-Visits TSP}:
Given a complete undirected graph $G$ with non-negative cost function $c:E\rightarrow \mathbb{Z}_{\geq 0}$ and requirements $r:V\rightarrow\mathbb{Z}_{\geq 1}$, find a vector $x\in\mathbb{Z}^{E}_{\geq 0}$ minimizing $c^Tx$ such that $d_x(v)=2r(v)$ for every $v\in V$, and $\supp(x)$ is connected.
From now on we use $\hat{r}=r(V)-|V|+1$.

The high-level idea of the algorithm is the following.
We first show that the set of integral vectors $\{x\in\mathbb{Z}^{E}_{\geq 0}:x(E)=r(V), ~\supp(x)\ \text{is connected}\}$ form the integral points of a base polymatroid.
We apply Corollary~\ref{thm:polym} to this base polymatroid to obtain a vector $x\in\mathbb{Z}^{E}_{\geq 0}$ with $c^Tx$ no more than the optimum, such that $d_x(v)\geq 2r(v)-1$ for $v\in V$.
Then we add a minimum-cost matching on the set of vertices of odd $d_x(v)$-value.
Finally, by shortcutting vertices with degree higher than prescribed, we obtain a tour that satisfies the requirements on the number of visits at every vertex.
\begin{lemma}
\label{lem:bdef}
  Let $b$ denote the following function defined on edge sets $F\subseteq E$:
  \begin{equation}
  \label{eq:b}
    b(F) = \begin{cases}
             |V(F)|-\comp(F)+\hat{r} & \text{if $F\neq\emptyset$,}\\
              0 & \text{otherwise.}
           \end{cases}
  \end{equation} 
  Then $b$ is a polymatroid function.
\end{lemma}

\begin{proof}
  By definition, $b(\emptyset)=0$ and $b$ is monotone increasing.
  It remains to show that $b$ is submodular.
  Let $X,Y\subseteq E$.
  The submodular inequality clearly holds if one of $X$ and $Y$ is empty. If none of $X$~and~$Y$ is empty then the submodular inequality follows from the fact that $|V(F)|-\comp(F)$ is the rank function of the graphical matroid. 
\end{proof}

Consider the base polymatroid $B(b)$ determined by the border function defined in \eqref{eq:b}.
Let us define the set $B=\{x\in\mathbb{Z}^{E}_{\geq 0}:x(E)=r(V), ~ \supp(x)\ \text{is connected}\}$.
\begin{lemma}
\label{lem:description}
  $B=B(b)\cap\mathbb{Z}^{E}_{\geq 0}$.
\end{lemma}
\begin{proof}
  Take an integral element $x\in B(b)$ and let $C\subseteq E$ be an arbitrary cut between $V_1$ and $V_2$ for some partition $V_1\uplus V_2$ of $V$.
  Then
  \begin{align*}
    x(C) {}&{}= x(E)-(x(E(V_1)\cup E(V_2)))\\
         {}&{}\geq |V|-1+\hat{r}-(|V_1|+|V_2|-\comp(E(V_1)\cup E(V_2))+\hat{r})\\
         {}&{}\geq 1,
  \end{align*}
  thus $\supp(x)$ is connected.
  As $x(E)=|V|-1+\hat{r}=r(V)$, we obtain $x\in B$, showing that $B(b)\subseteq B$.

  To see the other direction, take an element $x\in B$.
  As $\supp(x)$ is connected, $x(E - F)\geq\comp(F)+|V|-|V(F)|-1$ for every $F\subseteq E$.
  That is, 
  \begin{align*}
    x(F) {}&{}   = x(E) - x(E - F)\\
         {}&{}\leq r(V)-(|V-V(F)| +\comp(F)-1)\\
         {}&{}   = |V(F)|-\comp(F)+\hat{r},
  \end{align*}
  thus $x(F)\leq b(F)$.
  As $x(E) = r(V) = |V|-1+\hat{r}$, we obtain $x\in B(b)$, showing $B\subseteq B(b)$. 
\end{proof}

\begin{algorithm}[h!]
  \caption{A 1.5-approximation algorithm for the metric {\sc Many-Visits TSP}\label{alg:apx_matd}}
  \begin{algorithmic}[1]
    \Statex \textbf{Input:} A complete undirected graph $G$, costs $c:E\rightarrow\mathbb{R}_{\geq 0}$ satisfying the triangle inequality, requirements $r:V\rightarrow\mathbb{Z}_{\geq 1}$.
    \Statex \textbf{Output:} A tour that visits each $v \in V$ exactly $r(v)$ times.
    \State Construct the polymatroid $B(b)=(S,b)$, where $S:=E$ and $b$ is defined as in Equation~\eqref{eq:b}.\label{st:1} 
    \State  Construct a hypergraph $H=(S,\Eps)$ with $\Eps=\{\delta(v):v\in V\}$ and \label{st:2} 

      $\circ$ for every $\eps \in \Eps$ and $s\in\eps$, set $m_\eps(s)=2$ if $s$ is a self-loop and $m_\eps(s)=1$ otherwise, 

      $\circ$ for every $\eps \in \Eps$, set $f(\eps)=2 \cdot r(v)$, where $\eps=\delta(v)$. 

    \State Run Algorithm~\ref{alg:matd} with $B(b),c,H,f$ and the $m_\eps$'s as input.
      Let $z\in B(b)$ denote the output.\label{st:3}
    \State Calculate a minimum-cost matching $M$ with respect to $c$ on the vertices of $V$ with odd $d_z(v)$ values.\label{st:4}
    \State Determine a tour $T = \{C, \mu_C\}_{C \in \mathcal{C}}$ from $z$ and $\chi_M$.\label{st:5}
    \State Do shortcuts in $T$ and obtain a solution $T'$, such that $T'$ visits every city $v$ exactly $r(v)$ times (that is, $d_T(v) = r(v)$ for every vertex $v \in V$). \label{st:6}
    \State \textbf{return} $T'$.
  \end{algorithmic}
\end{algorithm} 

Our algorithm is presented as Algorithm~\ref{alg:apx_matd}. First, we construct a polymatroid $B$ and a hypergraph~$H$, such that their common ground set $S$ consists of the edges of the graph $G$ in our \textsc{Many-Visits TSP} instance. The border function $b$ of the polymatroid is defined in Equation~\eqref{eq:b}.

For each vertex $v$ of $G$, there is a hyperedge $\eps$ in the hypergraph that contains all edges of~$G$ incident to $v$, including the self-loop at $v$.
We set the multiplicities of an element $s \in S$ to~2 if it corresponds to a self-loop in $G$, and to 1 otherwise.
The motivation is that a self-loop contributes the the degree of a vertex by two, while a regular edge contributes to the degree of each of its endpoints by one.
Note that an element $s \in S$ is contained in exactly one hyperedge if it corresponds to a self-loop, and it is contained in exactly two hyperedges otherwise; therefore the total contribution of each edge adds up to two.

Now we are ready to prove Theorem~\ref{thm:tsp1}.
\begin{proof}[Proof of Theorem~\ref{thm:tsp1}]
  First, let us show that Algorithm~\ref{alg:apx_matd} provides a feasible solution for the given instance $(G,c,r)$ of the metric {\sc Many-Visits TSP}.

  By Lemma~\ref{lem:description}, the solution $z$ provided by Algorithm~\ref{alg:matd} in Step~\ref{st:2} is such that $c^T z\leq c(\TR^\star_{c,r})$, $z(E)=r(V)$ and $\supp(z)$ is connected.
  Note that in our case $\Delta=\max_{s\in S}\{\sum_{\eps \in \Eps:s\in \eps}m_\eps(s)\}\linebreak=2$.
  By Theorem~\ref{thm:matroid2}, $f(\eps)-1\leq\sum_{s \in \eps}m_\eps(s)z(s)$ for each $\eps \in \Eps$, and this inequality translates to $2\cdot r(v)-1\leq d_z(v)$ for every $v\in V$.
  That is, $z$ corresponds to a multigraph of cost at most $c(\TR^\star_{c,r})$ violating the degree prescriptions from below by at most one.
  Note that this means the total violation from above is at most $|V|-1$.
  In Step~\ref{st:3} we calculate a matching~$M$ that provides one extra degree to each odd-degree vertex.
  That is, in the union of the multigraph defined by~$z$ and $M$, every vertex~$v$ has an even degree.
  
  \paragraph{\textbf{Constructing a tour and shortcutting}}
  In Step~\ref{st:4}, we construct a \emph{compact representation} of a tour from the vector $z$ and matching $M$, and we denote it by~$T$.
  We use the algorithm described in Grigoriev~et~al.~\cite{Grigoriev2006}, which takes the edge multiplicities as input, and outputs a collection $\mathcal{C}$ of pairs~$(C, \mu_C)$.
  Here $C$ is a simple closed walk, and $\mu_C$ is the corresponding integer denoting the number of copies of the walk $C$ in $T$.
  From the pairs $(C, \mu_C)$ it is possible to construct an implicit order of the vertices the following way.
  
  Let us construct an auxiliary multigraph $A$ on the vertex set $V$ by taking the edges of each cycle~$C$ exactly once.
  Note that parallel edges are allowed in~$A$ if an edge appears in multiple cycles~$C$.
  Due to the construction, every vertex has an even degree in $A$, which means that there exist an Eulerian circuit in $A$.
  Moreover, there are $\mathcal O(n^2)$ distinct cycles~\cite{Grigoriev2006}, hence, the total number of edges in $A$ is $\mathcal O(n^3)$.
  Consequently, using Hierholzer's algorithm, we can compute an Eulerian circuit $\eta$ in $A$ in $\mathcal O(n^3)$ time~\cite{Hierholzer1873,Fleischer1991}.
  The circuit $\eta$ covers the edges of each cycle $C$ once.
  Now an implicit order of the vertices in the {\sc Many-Visits TSP} tour $T$ is the following.
  Traverse the vertices of the Eulerian circuit $\eta$ in order.
  Every time a vertex $u$ appears the first time, traverse all cycles~$C$ that contain the vertex $\mu_{C}-1$ times.
  Denote this circuit by $\eta'$.
  It is easy to see that the sequence~$\eta'$ is a sequence of vertices that uses the edges of each cycle $C$ exactly $\mu_C$ times, meaning this is a feasible sequence of the vertices in the tour~$T$.
  Moreover, the order itself takes polynomial space, as it is enough to store indices of $\mathcal O(n^3)$ vertices and $\mathcal O(n^2)$ cycles.
  
  Now let us consider the set $W$ of vertices $w$ that have more visits than $r(w)$ in the tour~$T$.
  Denote the surplus of visits of a vertex $w\in W$ by $\gamma(w) := d_T(w) - 2\cdot r(w)$.
  In Step~\ref{st:6}, we remove the last~$\gamma(w)$ occurrences of every vertex $w \in W$ from~$T$, by doing shortcuts: if an occurrence of~$w$ is preceded by~$u$ and superseded by~$v$ in $T$, replace the edges $uw$ and $wv$ by $uv$ in the sequence.
  This can be done by traversing the compact representation of $\eta'$ backwards, and removing the vertex $w$ from the last $\gamma(w)$ cycles $C^{(w)}_{r(w)-\gamma(w)+1}, \dots, C^{(w)}_{r(w)}$.
  As $\sum_w \gamma(w)$ can be bounded by~$\mathcal O(n)$, this operation makes $\mathcal O(n)$ new cycles, keeping the space required by the new sequence of vertices and cycles polynomial.
  Moreover, since the edge costs are metric, making shortcuts the way described above cannot increase the total cost of the edges in $T$.
  Finally, using a similar argument as in the algorithm of Christofides, the shortcutting does not make the tour disconnected.
  The resulting graph is therefore a tour $T'$ that visits every vertex~$v$ exactly $r(v)$ times, that is, a feasible solution for the instance $(G,c,r)$.
  
  \paragraph{\textbf{Cost and complexity}}
  The cost of the edges in $z$ is at most $c(\TR^\star_{c,r})$, and as the cost function satisfies the triangle inequality, the cost of the matching $M$ found in Step~\ref{st:3} is at most $c(\TR^\star_{c,1})/2$.
  Moreover, taking shortcuts at vertices does not increase the cost of the solution, hence the cost of the output is at most $c^T(z+\chi_M)\leq c(\TR^\star_{c,r})+c(\TR^\star_{c,1})/2=1.5\cdot c(\TR^\star_{c,r})$.

  Now we turn to the complexity analysis.
  All edge multiplicities during the algorithm are stored as integer numbers in binary, therefore the space needed of any variable representing multiplicities of edges can be bounded by $\mathcal O(n^2 \, \log\sum r(v))$.
  Steps~\ref{st:1}-\ref{st:2} can be performed in time that is polynomial in the input size.
  The function $b$ is defined in Lemma~\ref{lem:bdef} and can be computed efficiently.
  Therefore, according to Corollary~\ref{thm:polym}, the algorithm in Step~\ref{st:3} also runs in polynomial time.
  Step~\ref{st:4} can also be done in polynomial time~\cite{Grigoriev2006}, and the number of closed walks can be bounded by $\mathcal O(n^2)$.
  Moreover, the total surplus of degrees in $T$ is at most $n-1$, therefore the number of shortcutting operations is also bounded by $n$.
  This completes the proof.
\end{proof}

It is worth considering what Algorithm~\ref{alg:apx_matd} does when applied to the single-visit TSP, that is, when $r(v)=1$ for each $v\in V$.
The output of Algorithm~\ref{alg:matd} in Step~\ref{st:3} is a connected multigraph with $r(V)=n$ edges.
Note that the guarantee that each vertex $v$ has degree at least $2\cdot r(v)-1=1$ does not add anything as this already follows from connectivity.
Such a graph is basically the union of a spanning tree and a single edge (where the edge might be also part of the spanning tree, that is, in the solution having multiplicity 2); we call such a graph a \emph{1-tree}.
The rest of the algorithm mimics Christofides' algorithm: a minimum cost matching is added on the set of vertices of odd degree to get an Eulerian graph, and then a Hamiltonian circuit is formed by shortcutting repeated vertices in an Eulerian circuit.
That is, applying our algorithm to a single-visit TSP instance, it is almost identical to that of Christofides, except that instead of a spanning tree we start with a 1-tree.
However, the 1-tree we start with is not necessarily a cheapest one among all possible choices; we only know that its cost is at most the cost of the optimal single-visit TSP tour.

\section{Discussion}
\label{sec:discussion}
In this work we developed an approximation algorithm for the minimum-cost degree bounded g-polymatroid element problem with multiplicities.
The approximation algorithm yields a solution of cost at most the optimum, which violates the lower bounds only by a constant factor depending on the weighted maximum element frequency $\Delta$.
We then demonstrated the usefulness of our result by developing a polynomial-time $1.5$-approxi\-ma\-tion algorithm for the metric many-visits traveling salesman problem.
This way, we match the famous Christofides-Serdyukov bound for the single-visit TSP.

\bigskip
\noindent
{\small
\textbf{Acknowledgements.}
The authors are grateful to Rico Zenklusen for discussions on techniques to obtain a $1.5$-approximation for the metric version of the {\sc Many-Visits TSP}, and to Tam\'as Kir\'aly and Gyula Pap for their suggestions.
Krist\'of was supported by the J\'anos Bolyai Research Fellowship of the Hungarian Academy of Sciences and by the {\'U}NKP-19-4 New National Excellence Program of the Ministry for Innovation and Technology. Project no. NKFI-128673 has been implemented with the support provided from the National Research, Development and Innovation Fund of Hungary, financed under the FK\_18 funding scheme.
This research was supported by Thematic Excellence Programme, Industry and Digitization Subprogramme, NRDI Office, 2019.}

\bibliographystyle{abbrv}
\bibliography{mvtsp_apx}

\end{document}